\documentclass[runningheads,a4paper]{article}

\usepackage{here}
\usepackage[utf8]{inputenc}
\usepackage{graphicx}
\usepackage{latexsym}
\usepackage{amsmath}
\usepackage{lmodern,tabularx,booktabs}

\usepackage{amsthm}
\usepackage{amssymb}
\usepackage[ruled,vlined,linesnumbered]{algorithm2e}
\usepackage{enumerate}
\SetKwFor{ForEach}{for each}{do}{endfch}
\SetKwFor{ForAll}{for all}{do}{endfall}
\SetKw{KwTrue}{true}
\SetKw{KwFalse}{false}
\DontPrintSemicolon

\newcommand{\argmin}{\ensuremath{\mathop{\rm argmin} \limits}}

\newtheorem{definition}{Definition}
\newtheorem{theorem}{Theorem}

\newcommand{\rfig}[1]{Figure~\ref{#1}}

\newcommand{\rthm}[1]{Theorem~\ref{#1}}

\newcommand{\rsec}[1]{Section~\ref{#1}}

\newcommand{\rdef}[1]{Definition~\ref{#1}}

\title{Cyclic Oritatami Systems Cannot Fold\\Infinite Fractal Curves}

\author{Yo-Sub Han \and Hwee Kim}

\begin{document}
\maketitle
\begin{abstract}
RNA cotranscriptional folding is the phenomenon in which an RNA transcript folds upon itself while being synthesized out of a gene. The oritatami system (OS) is 
a computation model of this phenomenon, which lets its sequence (transcript) of beads (abstract molecules) fold cotranscriptionally by the interactions between beads according to the binding ruleset.
The OS is an useful computational model for predicting and simulating an RNA folding as well as constructing a biological structure.  
A fractal is an infinite pattern that is self-similar across different scales, and is an important structure in nature. Therefore, the fractal construction using self-assembly
is one of the most important problems. We focus on the problem of generating an infinite fractal instead of a partial finite fractal, which is much more challenging. 
We use a cyclic OS, which
has an infinite periodic transcript, to generate an infinite structure. We prove a negative result that it is impossible to make a Koch curve or a Minkowski curve, both of which are fractals, 
using a cyclic OS. We then establish sufficient conditions of infinite aperiodic curves that a cyclic OS cannot fold.

\end{abstract}

\section{Introduction}

Self-assembly is the process where smaller components---usually molecules---autonomously assemble and form a larger complex structure. Self-assembly plays an important role in constructing biological structures and high polymers~\cite{WhitesidesB02}. 
One well-known mathematical model of the self-assembly phenomenon is the abstract tile assembly model (aTAM)~\cite{Winfree98}. Recently, Geary et al.~\cite{GearyMSS16} proposed a new computation model called the oritatami system (OS) that simulates the cotranscriptional self-assembly based on the experimental RNA transcription called RNA origami~\cite{GearyRA14}. In general, the OS assumes that a sequence of molecules is transcribed linearly, and predicts the geometric shape of the autonomous folding of the sequence based on the reaction rate of the folding. The OS consists of a sequence of beads (which is the transcript) and a set of rules for possible intermolecular reactions between beads. For each bead in the sequence, the system takes a lookahead of a few upcoming beads and determines the best location of the bead that maximizes the number of possible interactions from the lookahead. The lookahead represents the reaction rate of the cotranscriptional folding and the number of interactions represents the energy level (See \rfig{fig_osex2} for the analogy between RNA origami and oritatami system). 

\begin{figure}[tb]
	\centering
	\includegraphics[width=1.0\textwidth, bb=0 0 550 300]{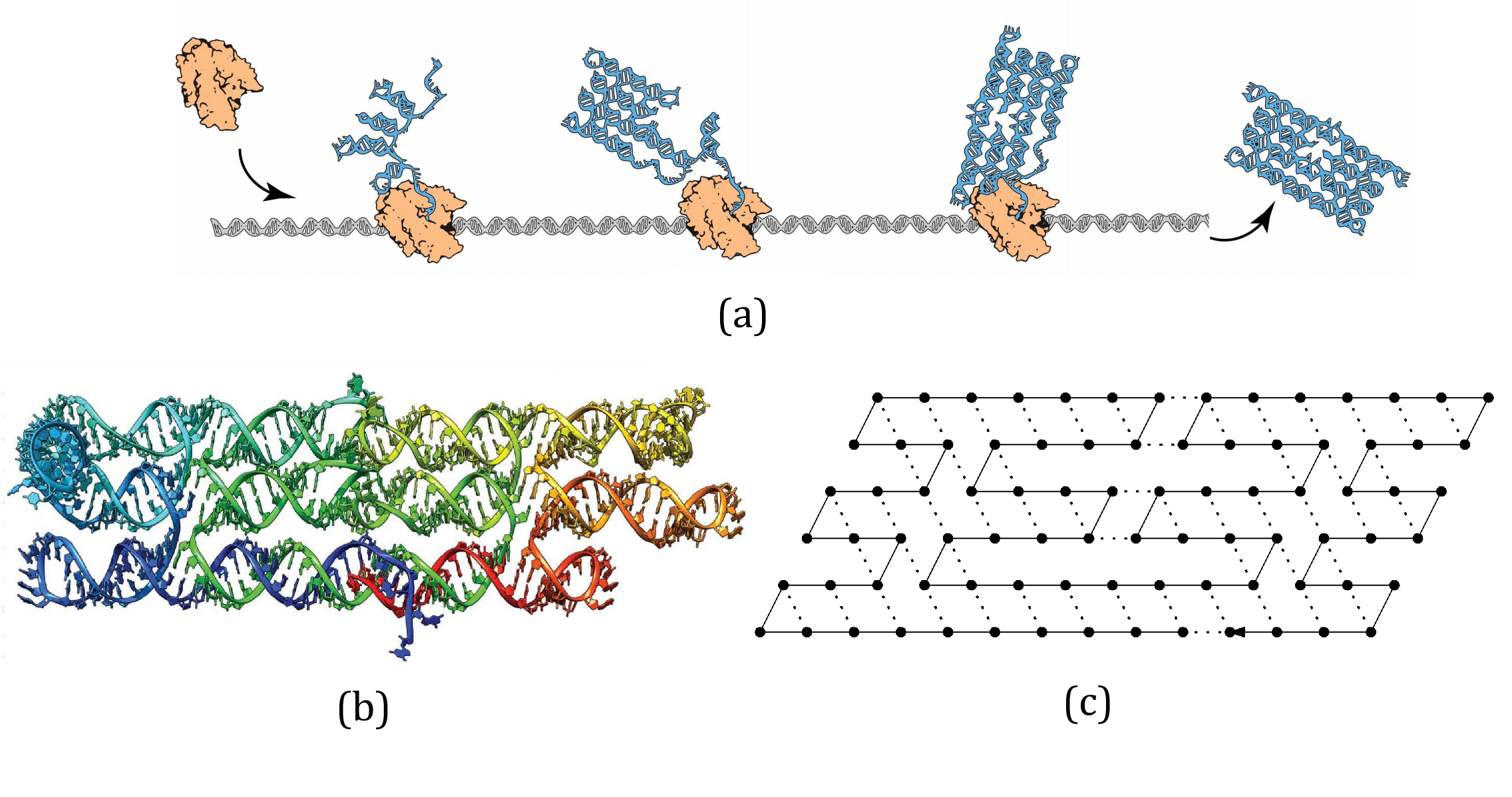}
	\caption{The motivation of the oritatami system. (a) An illustration of an RNA Origami~\cite{GearyMSS16}, which transcribes an RNA strand that self-assembles. (b) The product of an RNA Origami. (c) Abstraction of the product in the oritatami system.}
	\label{fig_osex}
\end{figure}

\begin{figure}[htb]
	\centering
	\includegraphics[width=1.0\columnwidth]{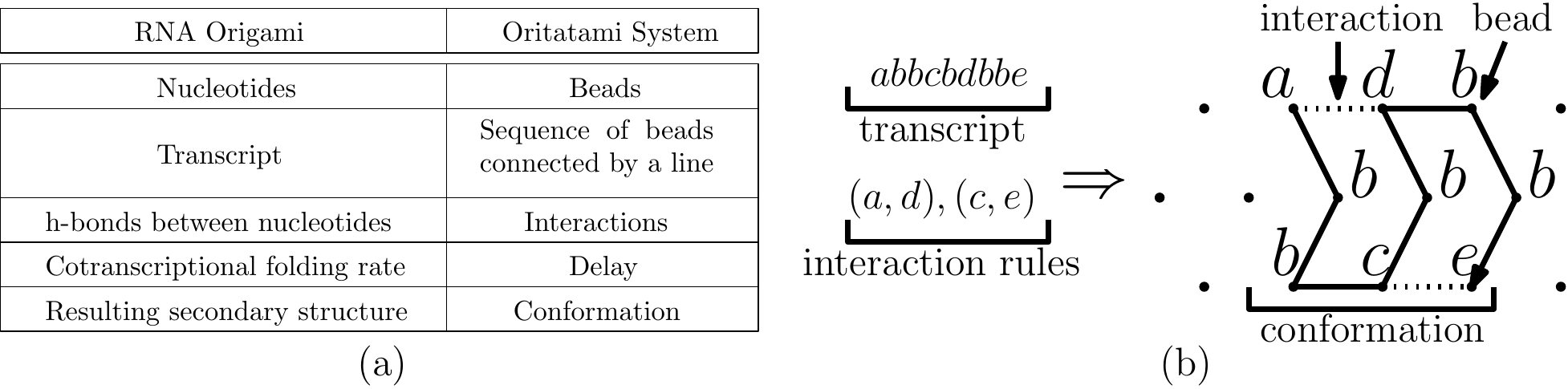}
	\caption{(a) Analogy between RNA origami and oritatami system. (b) Visualization of oritatami system and its terms.}
	\label{fig_osex2}
\end{figure}

The OS is a computation model based on geometric aspects, and, thus, it is important to design and analyze the OS 
in both computational and geometric perspectives. From the computational point of view, the OS is proved to be Turing complete using a cyclic tag system simulator~\cite{GearyMSS16}. Designs for other computational problems such as binary counting~\cite{GearyMSS15} and  Boolean formula simulation~\cite{HanKOS18} are also established. 
Researchers also proved a few decision properties and proposed various optimization methods for the OS design~\cite{HanK17,HanKRS17,OtaS17}. 
From the geometric point of view, Rogers and Seki~\cite{RogersS17} proved the decidability of geometric structure constructions based on the delay. Recently, Masuda et al.~\cite{MasudaSU18} proposed how to construct a finite Heighway dragon using a cyclic OS, which has a periodic transcript.

A fractal is an infinite pattern that is self-similar across different scales, and is an important structure in nature. The construction of fractals is one of the most important topics in both geometric computation models~\cite{HendricksOPRT18,HendricksO17,LathropLS09} and experiments~\cite{TesoroA16,TikhomirovPQ17}. For constructing fractals by the OS, we need a few assumptions.
\begin{enumerate}
	\item Since the OS transcribes an RNA single strand, it is natural to consider fractal curves, which are infinite sequences of segments (and points). 
	\item We only consider a deterministic OS (that only folds into a unique conformation). Note that it is trivial to make a nondeterministic OS that may fold into several different structures including a target structure. 
	\item We consider an infinite fractal construction. Masuda et al.~\cite{MasudaSU18} proposed how to construct a finite fractal by implementing a counter and an automaton periodically inside the transcript. This approach can be used to construct an arbitrary long fractal but not an infinite one, since the counter is finite.
	\item Since we need an infinite transcript with a formal finite representation, it is a natural choice to use a cyclic OS that has an infinitely repeated transcript. Researchers already 
	used cyclic OSs to construct conformations that can grow infinitely~\cite{GearyMSS15,GearyMSS16}.
\end{enumerate}

Although fractals have self-similarity, they do not have repeated structures with any fixed period. In aTAM, it is known that we can assemble a certain type of fractals infinitely even with limited tile types~\cite{LathropLS09}. In OS, we claim that under some reasonable assumptions, it is impossible to fold some infinite fractals with a cyclic OS. We first define the folding of the curve by an OS as mapping segments and points of the curve to sets of points on the OS triangular lattice. 

Before presenting our main contributions on the impossibility of OS
folding, we start with two well-known fractal curves (Koch and Minkowski curves) with respect to the OS folding and provide a brief idea on the impossibility results. We, in particular, prove that regardless of the delay and the period, it is impossible to fold Koch and Minkowski curves in Sections~\ref{sec_koch} and \ref{sec_min}, respectively. Then we
generalize this impossibility to infinite aperiodic curves and
establish sufficient conditions together with main contributions in
\rsec{sec_imp}.

We then examine two well-known fractal curves (Koch and Minkowski curves), and prove that regardless of the delay and the period, it is impossible to fold such curves under our assumptions. We expand the result to infinite aperiodic curves and establish sufficient conditions to prove impossibility of the folding.
	
\section{Preliminaries}
\label{sec_pre}

Let $w = a_1 a_2 \cdots a_n$ be a string over $\Sigma$ for some integer~$n$ and bead types~$a_1, \ldots, a_n$\\$ \in \Sigma$. 
The \emph{length}~$|w|$ of $w$ is $n$. 
For two indices~$i,j$ with $1\le i\le j\le n$, we let $w[i,j]$ be the substring~$a_ia_{i{+}1}\cdots a_{j{-}1}a_j$; we use $w[i]$ to denote $w[i,i]$. 
We use $w^n$ to denote the concatenation of $n$ copies of $w$.

Oritatami systems operate on the triangular lattice~$\Lambda_o$ with the vertex set~$V$ and the edge set~$E$. A \emph{configuration} is a triple~$(P, w, H)$ of a directed path~$P$ in $\Lambda_o$, $w \in \Sigma^* \cup \Sigma^\omega$, and a set $H \subseteq \left\{(i, j) \bigm| 1 \le i, i+2 \le j, \{P[i], P[j]\} \in E \right\}$ of interactions.
This is to be interpreted as the sequence~$w$ being folded while its $i$-th bead $w[i]$ is placed on the $i$-th point $P[i]$ along the path and there is an interaction between the $i$-th and $j$-th beads if and only if $(i, j) \in H$.
Configurations~$(P_1, w_1, H_1)$ and $(P_2, w_2, H_2)$ are \emph{congruent} provided $w_1 = w_2$, $H_1 = H_2$, and $P_1$ can be transformed into $P_2$ by a combination of a translation, a reflection, and rotations by $60$~degrees.
The set of all configurations congruent to a configuration~$(P, w, H)$ is called the \emph{conformation} of the configuration and denoted by $C=[(P, w, H)]$.
We call $w$ a \emph{primary structure} of $C$.

A ruleset~$\mathcal{H}\subseteq \Sigma\times\Sigma$ is a symmetric relation specifying between which bead types can form an interaction. 
An interaction~$(i, j) \in H$ is \emph{valid with respect to $\mathcal{H}$}, or simply \emph{$\mathcal{H}$-valid}, if $(w[i], w[j]) \in \mathcal{H}$.
We say that a conformation~$C$ is \emph{$\mathcal{H}$-valid} if all of its interactions are $\mathcal{H}$-valid.
For an integer~$\alpha \ge 1$, $C$ is \emph{of arity $\alpha$} if the maximum number of interactions per bead is $\alpha$, that is, if for any $k \ge 1$, $\bigl|\{i \mid (i, k) \in H \} \bigr| + \bigl|\{j \mid (k, j) \in H\} \bigr| \le \alpha$ and this inequality holds as an equation for some $k$.
By $\mathcal{C}_{\le \alpha}$, we denote the set of all conformations of arity at most $\alpha$.

Oritatami systems grow conformations by elongating them under their own ruleset.
For a finite conformation~$C_1$, we say that a finite conformation~$C_2$ is an \emph{elongation} of $C_1$ by a bead~$b \in \Sigma$ under a ruleset~$\mathcal{H}$, written as $C_1 \stackrel{\mathcal{H}}{\to}_b C_2$, if there exists a configuration~$(P, w, H)$ of $C_1$ such that $C_2$ includes a configuration~$(P\cdot p, w\cdot b, H \cup H')$, where $p \in V$ is a point not in~$P$ and $H' \subseteq \left\{(i, |P|{+}1) \bigm| 1 \le i \le |P|-1, \{P[i], p\} \in E, (w[i], b) \in \mathcal{H} \right\}$.
This operation is recursively extended to the elongation by a finite sequence of beads as follows:
For any conformation~$C$, $C \stackrel{\mathcal{H}}{\to}_\lambda^* C$; and
for a finite sequence of beads~$w$ and a bead~$b$, a conformation~$C_1$ is elongated to a conformation $C_2$ by $w\cdot b$, written as $C_1 \stackrel{\mathcal{H}}{\to}_{w\cdot b}^* C_2$, if there is a conformation~$C'$ that satisfies $C_1 \stackrel{\mathcal{H}}{\to}_w^* C'$ and $C' \stackrel{\mathcal{H}}{\to}_b C_2$.

An \emph{oritatami system} (OS) is a $6$-tuple~$\Xi = (\Sigma, w, \mathcal{H}, \delta, \alpha, C_\sigma=[(P_\sigma,w_\sigma,H_\sigma)])$, where $\mathcal{H}$ is a \emph{ruleset}, $\delta \ge 1$ is a \emph{delay}, and $C_\sigma$ is an $\mathcal{H}$-valid initial \emph{seed} conformation of arity at most $\alpha$, upon which its \emph{transcript} $w \in \Sigma^* \cup \Sigma^\omega$ is to be folded by stabilizing beads of $w$ one at a time and minimize energy collaboratively with the succeeding $\delta-1$ nascent beads.
The energy of a conformation~$C = [(P, w, H)]$ is $U(C)=-|H|$; namely, the more interactions a conformation has, the more stable it becomes.
The set~$\mathcal{F}(\Xi)$ of conformations \emph{foldable} by this system is recursively defined as follows: The seed~$C_\sigma$ is in $\mathcal{F}(\Xi)$; and provided that an elongation~$C_i$ of $C_\sigma$ by the prefix~$w[1:i]$ be foldable (i.e., $C_0 = C_\sigma$), its further elongation~$C_{i{+}1}$ by the next bead $w[i{+}1]$ is foldable if
\begin{equation}\label{eq:fold}
C_{i{+}1} \in \argmin_{\substack{C \in \mathcal{C}_{\le \alpha} \ \mbox{s.t.} \\ C_i \stackrel{\mathcal{H}}{\to}_{w[i{+}1]} C}} \min \left\{U(C') \Bigm| C \stackrel{\mathcal{H}}{\to}_{w[i{+}2:i{+}k]}^* C', k \le \delta, C' \in \mathcal{C}_{\le \alpha}
\right\}.
\end{equation}
Once we have $C_{i{+}1}$, we say that the bead~$w[i{+}1]$ and its interactions are \emph{stabilized} according to $C_{i{+}1}$.
A conformation foldable by $\Xi$ is \emph{terminal} if none of its elongations is foldable by $\Xi$. An OS is \emph{deterministic} if, for all $i$, there exists at most one $C_{i{+}1}$ that satisfies \eqref{eq:fold}. Namely, a deterministic OS folds into a unique terminal conformation. An OS is \emph{cyclic} if its transcript~$w=w_o^\omega$ is repetition of a string~$w_o$. We say that the OS has a period~$|w_o|$.

\begin{figure}[htb]
	\centering
	\includegraphics[width=1.0\textwidth]{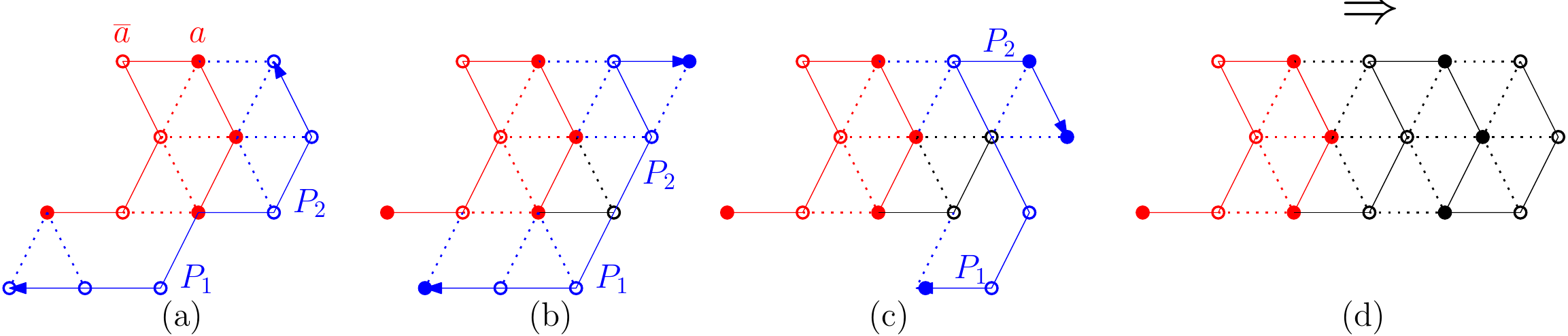}
	\caption{An example OS with delay~$3$ and arity~$4$. The seed is colored in red, elongations are colored in blue, and the stabilized beads and interactions are colored in black.}
	\label{fig_glider}
\end{figure}

\rfig{fig_glider} illustrates an example of an OS with delay~$3$, arity~$4$, ruleset~$\{(a,\overline{a})\}$ and transcript~$w=\overline{a}\overline{a}\overline{a}aaa\overline{a}\overline{a}\overline{a}$; in (a), the system tries to stabilize the first bead~$\overline{a}$ of the transcript, and the elongation~$P_1$ gives $2$~interactions, while the elongation~$P_2$ gives $4$~interactions, which is the most stable one. Thus, the first bead~$\overline{a}$ is stabilized according to the location in $P_2$. In (b) and (c), $P_2$ is the most stable elongation and $\overline{a}$'s are stabilized according to $P_2$. As a result, the terminal conformation is given as in (d). Note that the system grows the terminal conformation straight without external interactions, and we can use $w=(\overline{a}\overline{a}\overline{a}aaa)^\omega$ to fold an infinite periodic conformation. This example is called a \emph{glider}~\cite{GearyMSS15}.

The bead stabilization in OS is a local optimization of finding the best position of the bead using the next $\delta$~beads. Thus, the stabilization of a bead~$w[i]$ in a delay-$\delta$ OS is not affected by any bead whose distance from $w[i{-}1]$ is 
greater than $\delta+1$. On the triangular lattice, we can draw a hexagonal border of radius~$\delta+1$ from $w[i{-}1]$ to identify the set of points that may affect the stabilization of $w[i]$. 
While stabilizing a bead~$w[i]$, we define the \emph{event horizon} of $w[i]$ to be a partial conformation within this hexagon. Namely, the event horizon is the context used to stabilize $w[i]$. Thus, if two beads $w[i]$ and $w[j]$ have the same event horizon, then $w[i]$ and $w[j]$ are
stabilized congruently (See \rfig{fig_eh}). We define an event horizon for a partial transcript~$w[i,j]$ to be the union of event horizons of beads in $w[i,j]$ while stabilizing $w[i]$.

\begin{figure}[htb]
	\centering
	\includegraphics[width=0.7\textwidth]{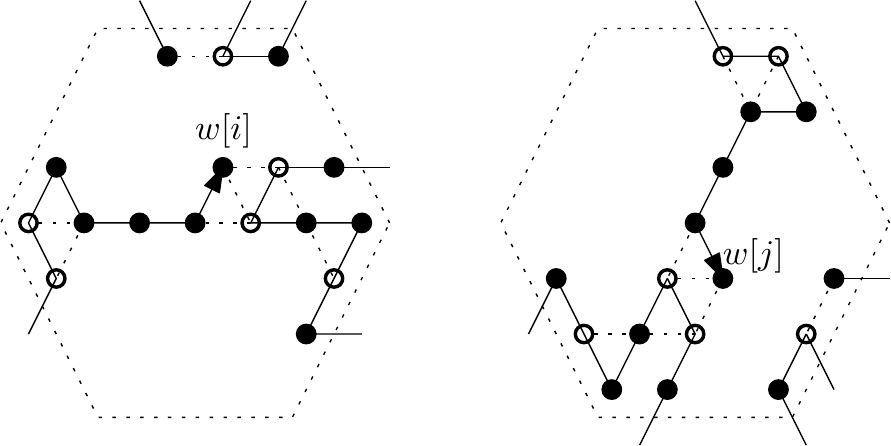}
	\caption{Two same event horizons when $\delta=2$ and we have two bead types (black and white circles). The current bead, pointed by an arrow, is stabilized in the same way in both event horizons.}
	\label{fig_eh}
\end{figure}

An L-system is a parallel rewriting system and its recursive nature makes the system easy to describe fractal-like structures~\cite{RozenburgS80}. 
An \emph{L-system} is defined as $G=(V,C,\omega,P)$, where
\begin{itemize}
	\item $V$ is the set of variables that can be replaced by production rules,
	\item $C$ is the set of constants that do not get replaced,
	\item $\omega\in(V\cup C)^*$ is the axiom, the initial string, and
	\item $P\subseteq V\times (V\cup C)^*$ is the set of production rules defining rewriting of variables.
\end{itemize}

The system starts with $\omega$, and as many rules as possible 
are applied simultaneously for each iteration. With graphical semantics on variables and constants, the L-system is often used to represent self-similar fractals. In this paper, we assume that curves are represented by strings, whose characters represent turns and unit segments. Then, an infinite curve is periodic if there exists a periodic string representation with a fixed finite period, and is aperiodic otherwise. Note that all fractal infinite curves are aperiodic.

\section{Impossibility of the Infinite Koch Curve}
\label{sec_koch}

We start with one example of infinite fractal curves---the Koch curve. The Koch curve can be constructed by starting with a segment, then recursively altering each line segment as follows:
\begin{enumerate}
	\item Divide the line segment into three segments of equal length.
	\item Draw an equilateral triangle that has the middle segment from step 1 as its base and points outward.
	\item Remove the line segment that is the base of the triangle from step 2.
\end{enumerate}
Using the L-system, the Koch curve can be encoded as follows:
\begin{itemize}
	\item Variable: $F$
	\item Constants: $+,-$
	\item Axiom: $F$
	\item Production Rule: $F\to F+F-F+F$,
\end{itemize}
where $F$ denotes a segment, $-$ denotes $120^\circ$ right turn and $+$ denotes $60^\circ$ left turn. \rfig{fig_koch_ex} illustrates the Koch curve after three iterations.

\begin{figure}[htb]
	\centering
	\includegraphics[scale=0.4]{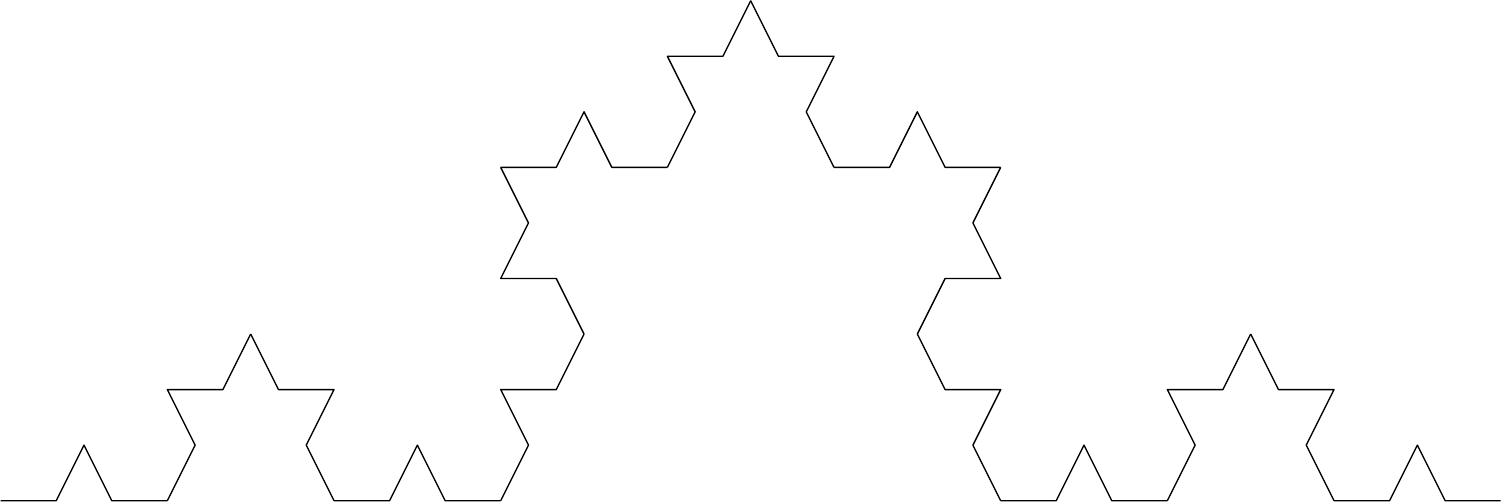}
	\caption{The Koch curve after three iterations}
	\label{fig_koch_ex}
\end{figure}

For this specific curve, we first make a few assumptions about the curve representation by a cyclic OS and the OS itself. We prove that it is impossible to draw an infinite Koch curve using a cyclic OS under our assumptions.

Let a \emph{shape} be a set of points on the triangular lattice~$\Lambda_o$, whose grid graph is connected. A curve can be represented as a sequence of alternating points and segments. We say that a sequence of shapes represents a curve if there exists one-to-one correspondence between shapes and alternating points and segments, and a point and a segment should be adjacent on the curve if and only if two shapes corresponding to them are adjacent. We formally define the drawing of the curve by a deterministic OS as follows:

\begin{definition}
	\label{def_draw}
	Given a (possibly infinite) curve on the plane and the (possibly infinite) sequence~$(S_k)$ of shapes that represents the curve, we say that a deterministic OS \emph{draws} the curve if the following condition holds: There exists a (possibly infinite) sequence~$(i_k)$ of indices that corresponds to the sequence of shapes, where, for all $k$'s, there exists a partial configuration for $w[i_{k{-}1}{+}1:i_k]$ that folds within $S_k$. We say that the OS \emph{covers} $S_k$ with the partial transcript~$w[i_{k{-}1}{+}1:i_k]$.
\end{definition}

\rfig{fig_draw} shows two examples of curve drawing by an OS. Here, the target curve is represented by three shapes~$(S_1,S_2,S_3)$. By \rdef{def_draw}, an OS draws the curve in \rfig{fig_draw}~(a) but does not in \rfig{fig_draw}~(b) since the partial configuration for $w[3:11]$ is not within $S_2$. Note that shapes limit paths of conformations, and it is not necessary to fill all points in the shape with the conformation.

From the design perspective, it is crucial to assume this locality of partial configuration. OS designs are usually modular~\cite{GearyMSS15,GearyMSS16,HanKOS18,MasudaSU18}---a partial transcript is folded locally under a controlled context, and we connect these partial transcripts to perform complex computations. Especially, for an infinite transcript, it becomes almost impossible to remove unintended interferences without this locality. Previous cyclic OSs such as a binary counter~\cite{GearyMSS15} or a cyclic tag system~\cite{GearyMSS16} follow this assumption. Remark that the drawing in \rdef{def_draw} is general in the sense that it does not restrict the OS to be infinite or cyclic, and the curve to be infinite.

\begin{figure}[htb]
	\centering
	\includegraphics[width=1.0\textwidth]{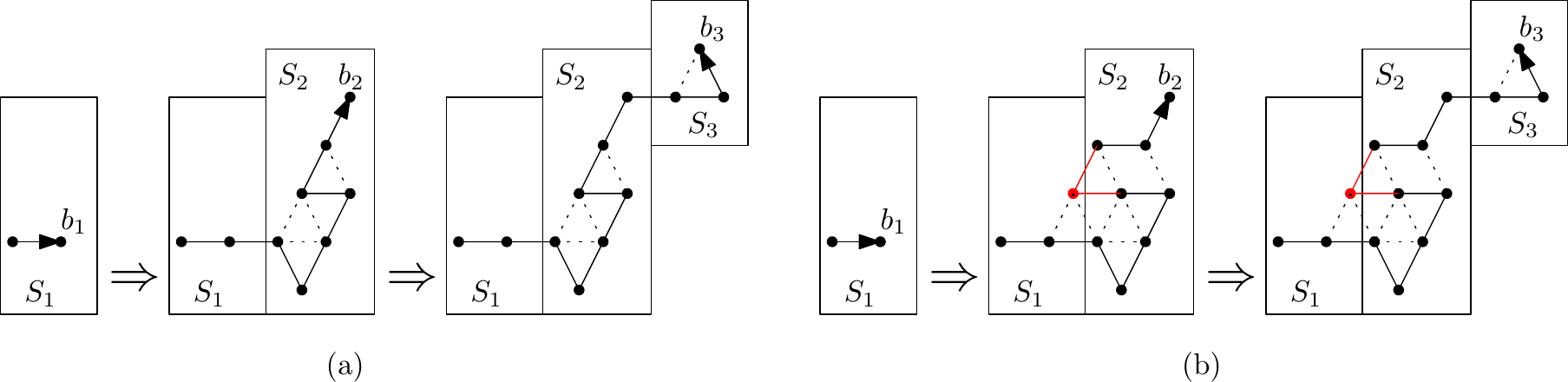}
	\caption{(a) OS draws the target curve and (b) OS does not draw the target curve by \rdef{def_draw}}
	\label{fig_draw}
\end{figure}

We assume the followings to draw the Koch curve:
\begin{enumerate}
	\item The Koch curve consists of an infinite sequence of alternating points and segments on the triangular lattice~$\Lambda_c$, which is 
	with vertical rows with unit triangles pointing left and right. 
	We use a hexagon~$S_p$ of side length~$d$ to represent a point on the curve. For a segment on the curve, we use a shape~$S_l$ of $d$~points ($l+1$~rows in total) and $d+1$~points ($l$~rows in total) in alternative positions which are orthogonal to the direction of the segment (See \rfig{fig_koch_shape} (a)). 
	The OS starts with covering the first $S_p$, and denote the $i$th $S_p$ ($S_l$) by $S_p[i]$ ($S_l[i]$). 
	\item We use constant numbers of beads for segments or points---$p_p$~beads in $S_p$, and $p_l$~beads in $S_l$. This assumption is reasonable in the modular design of the OS.
	\item The period~$p$ of the OS is $p=p_p+p_l$.
\end{enumerate}

\begin{figure}[htb]
	\centering
	\includegraphics[width=0.8\textwidth]{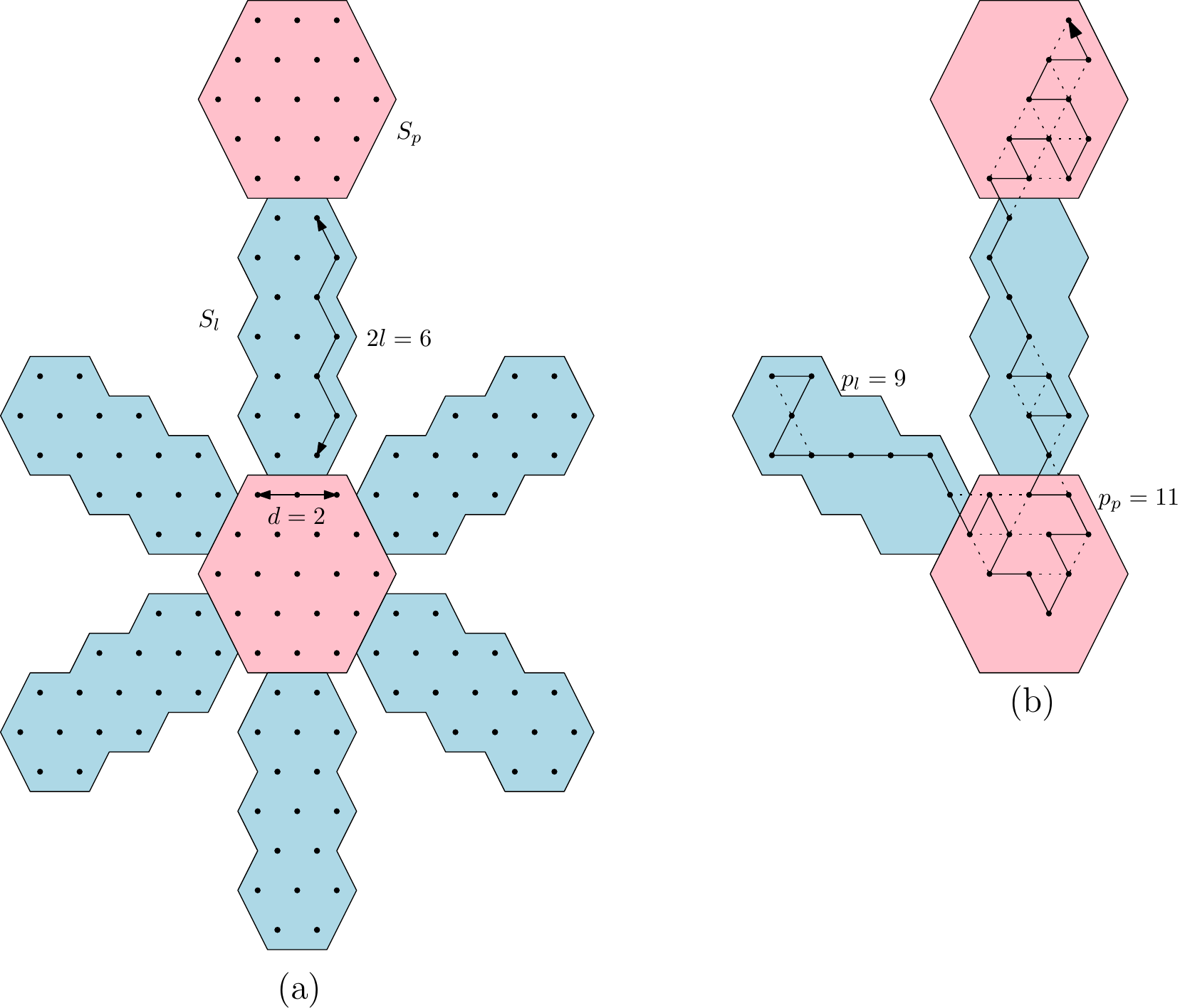}
	\caption{(a) Shapes used to draw the Koch curve (b) An example of an OS with $p_l=9$ and $p_p=11$}
	\label{fig_koch_shape}
\end{figure}

\rfig{fig_koch_shape} shows an example of shapes that can be used to draw the Koch curve, and a part of an OS that draws the curve, following the above assumptions. In \rfig{fig_koch_shape} (a), $S_p$'s are drawn in red, and $S_l$'s are drawn in blue, where $d=2$ and $l=3$. In \rfig{fig_koch_shape} (b), from assumption (2), the number of beads for segments or points are constant, even if there are different paths in different $S_p$'s or $S_l$'s. For example, paths in $S_p$ use $11$~beads and paths in $S_l$ use $9$~beads.

Under these assumptions, we claim the following theorem.
\begin{theorem}
	\label{thm_koch}
	There is no deterministic OS that can draw the Koch curve.
\end{theorem}
\begin{proof}
	Assume that there exists a deterministic cyclic OS~$\Xi$ with delay~$\delta$ that draws the Koch curve. First, we assume that $\delta<3d+3l+2$, which is one less than the distance between two $S_p$'s on $\Lambda_o$---apart by two unit distances on $\Lambda_c$. If the delay~$\delta$ has an upper bound, we denote the event horizon for the maximum delay as the ``maximum'' event horizon, and omit the term maximum if the context is clear. For convenience, we use $S_{pl}[i]$ to denote the shape resulted from connecting $S_p[i]$ and $S_l[i]$.

	We consider bead stabilization in $S_{pl}[i]$. We have an event horizon~$E(i)$ that is used to fold the conformation in $S_{pl}[i]$. Due to the delay upper bound, all $S_l$'s that overlap with $E(i)$ are at most three unit distances apart from $S_p[i]$, and all $S_p$'s that overlap with $E(i)$ are at most two unit distances apart from $S_p[i]$ (See \rfig{fig_koch_eh}).
	
	\begin{figure}[htb]
		\centering
		\includegraphics[width=1.0\textwidth]{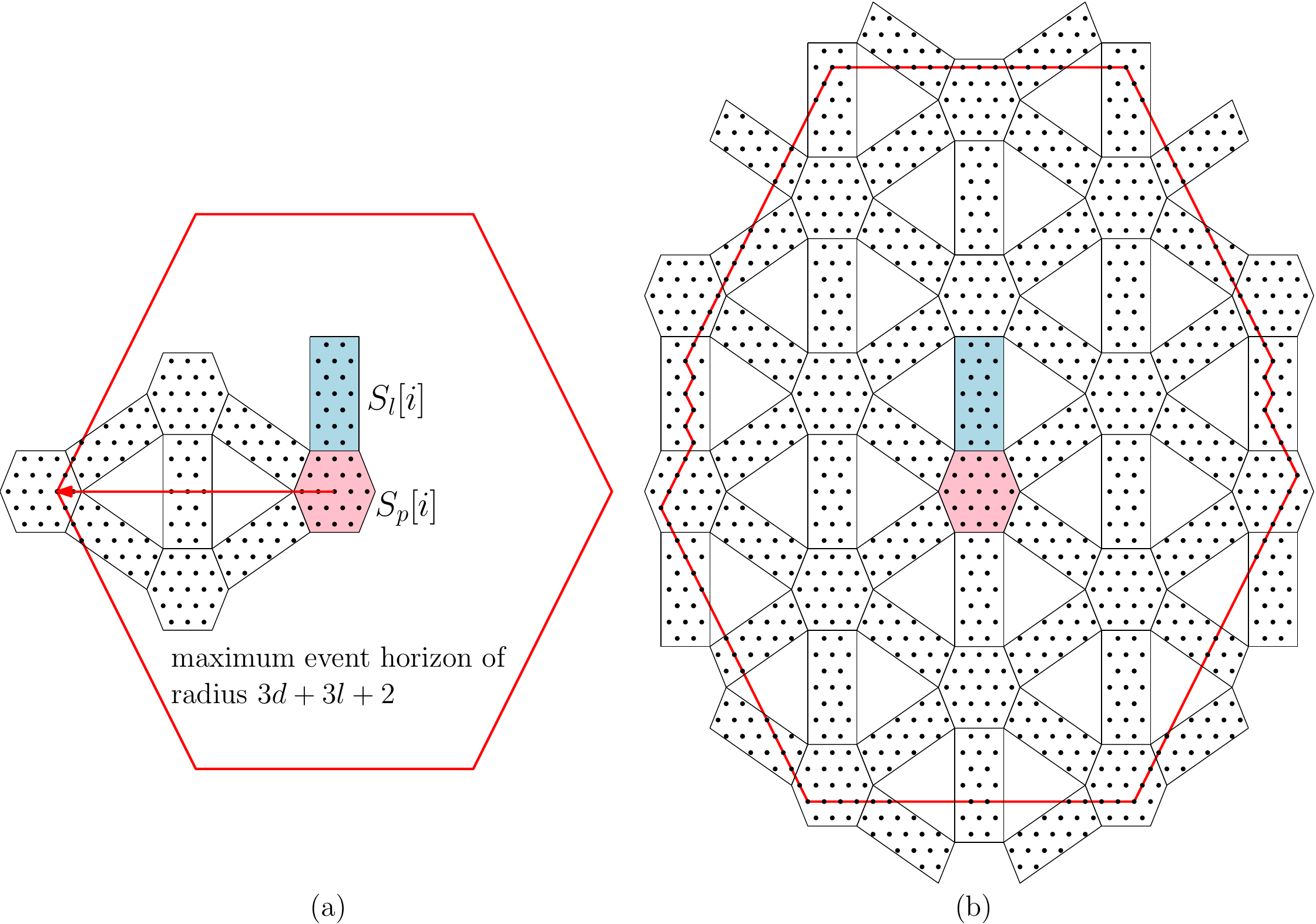}
		\caption{(a) The boundary of a event horizon is represented by a red hexagon. (b) The boundary of the event horizon~$E(i)$ for beads in $S_{pl}[i]$ is represented by a red polygon. $S_l$'s and $S_p$'s that overlap with $E(i)$ are depicted.}
		\label{fig_koch_eh}
	\end{figure}
	
	Since the Koch curve does not touch itself, if a point in $\Lambda_c$ is adjacent to two segments of the curve, there is no segment other than two that are adjacent to the point. Moreover, due to the self-similarity of the curve, the same statement holds for any scale of the power of $3$---for a point~$q_0$ in $\Lambda_c$, if there exist two points~$q_1$ and $q_2$ on the curve that are $3^n$ unit distances straight away from $q_0$, then there is no segment within $3^n$ unit distances from $q_0$, other than segments on the curve from $q_1$ and $q_2$ (See \rfig{fig_touch}).
	
	\begin{figure}[htb]
		\centering
		\includegraphics[scale=0.9]{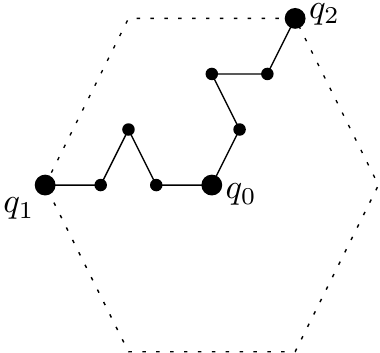}
		\caption{An example of $q_0$, $q_1$ and $q_2$. Points~$q_1$ and $q_2$ are $3$ unit distances straight away from $q_0$. Aside from curves from $q_1$ and $q_2$, there is no segment within $3$ unit distances from $q_0$ (depicted by a dotted hexagon).}
		\label{fig_touch}
	\end{figure}
	
	We have nine different cases for the beads within $S_p$'s and $S_l$'s that overlap with $E(i)$ as shown in \rfig{fig_koch_01}. 
	In each case, we first transcribe a sequence of blue shapes (if they exist), and then a sequence of red shapes. Shapes~$S_p[i]$ and $S_l[i]$ are distinguished by points within the shape. The thick hexagons represent $q_0$'s of distance $3$. We can observe that in all cases, all $S_p$'s and $S_l$'s are at most three unit distances apart from one of $q_0$'s, and they are empty except the red and blue shapes. Moreover, in all cases, all beads within $E(i)$ are from $S_p[i{-}4]$ to $S_l[i{-}1]$, a consecutive sequence of shapes before $S_l[i]$. 
	In other words,	a partial conformation in $S_{pl}[i]$ is dependent to beads in the sequence of shapes from $S_p[i{-}4]$ to $S_l[i{-}1]$.
	
	Since the period of the OS is $p_1+p_2$, we have exactly the same partial transcripts that fold within all $S_{pl}$'s, and, therefore, 
	having the same path yields the same partial conformation. The upper bound for the number of possible paths within $S_p$ ($S_l$) is $5^{p_1}$ ($5^{p_2}$). We have beads from $S_p[i{-}4]$ to $S_l[i{-}1]$ that determine the partial conformation in $S_{pl}[i]$.
	Thus, we have $i$ and $j$ such that $1\le i<j\le 1+5^{4p_l+4p_p}$, and $S_{pl}[i]$ and $S_{pl}[j]$ have exactly the same conformation. Moreover, since beads from $S_p[i{-}4]$ to $S_l[i{-}1]$ are consecutively transcribed, $S_{pl}[i]$ and $S_{pl}[j]$ result in a periodic sequence of segment turns of length~$j-i$ (See \rfig{fig_koch1}). Since the Koch curve is aperiodic, we know that the first assumption $\delta<3l+3d+2$ is wrong.
	\newpage
	\begin{figure}[h!]
		\centering
		\includegraphics[scale=0.55]{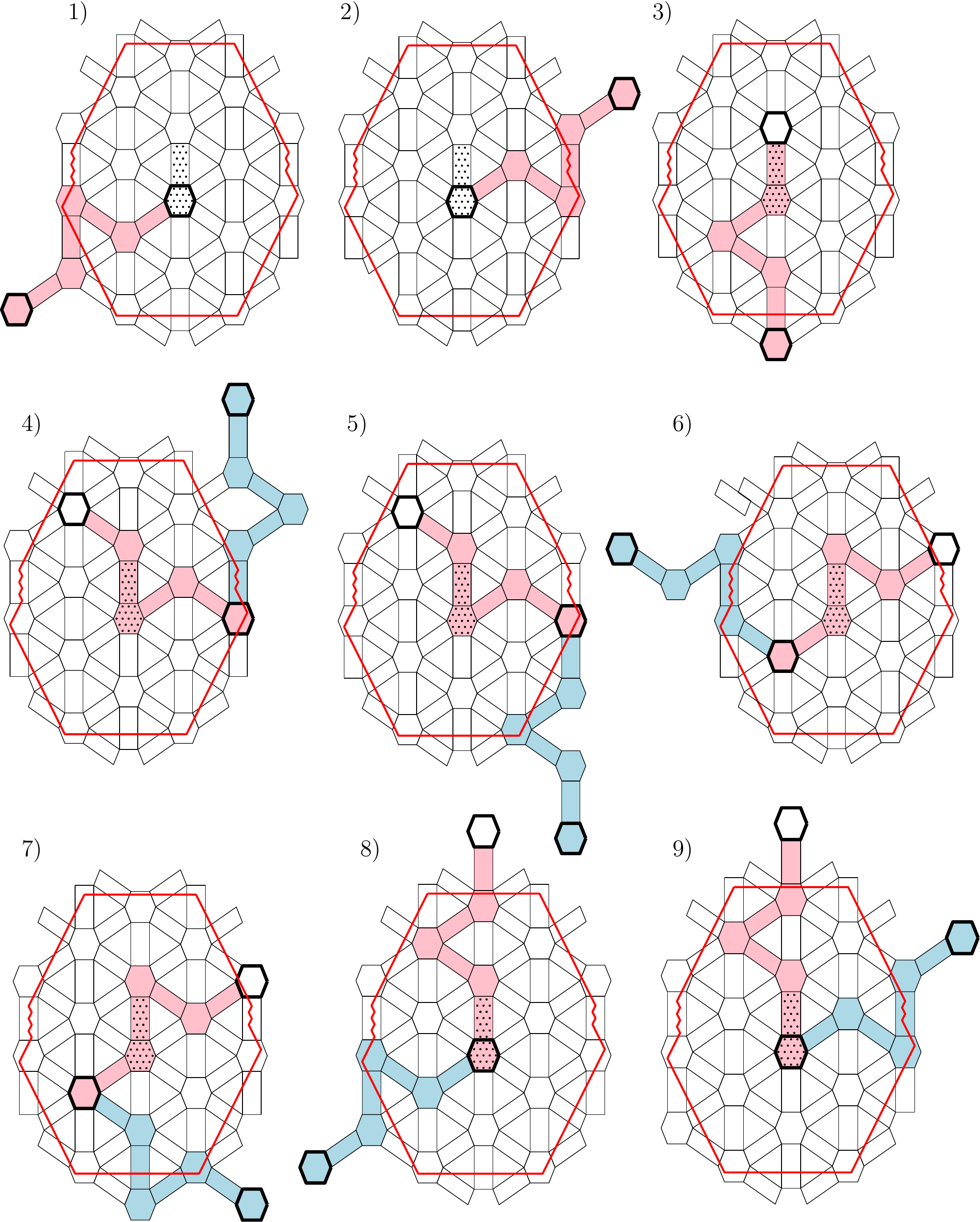}
		\caption{Nine cases of the beads in the union of event horizons. }
		\label{fig_koch_01}
	\end{figure}

	\begin{figure}[htb]
		\centering
		\includegraphics[width=1.0\textwidth]{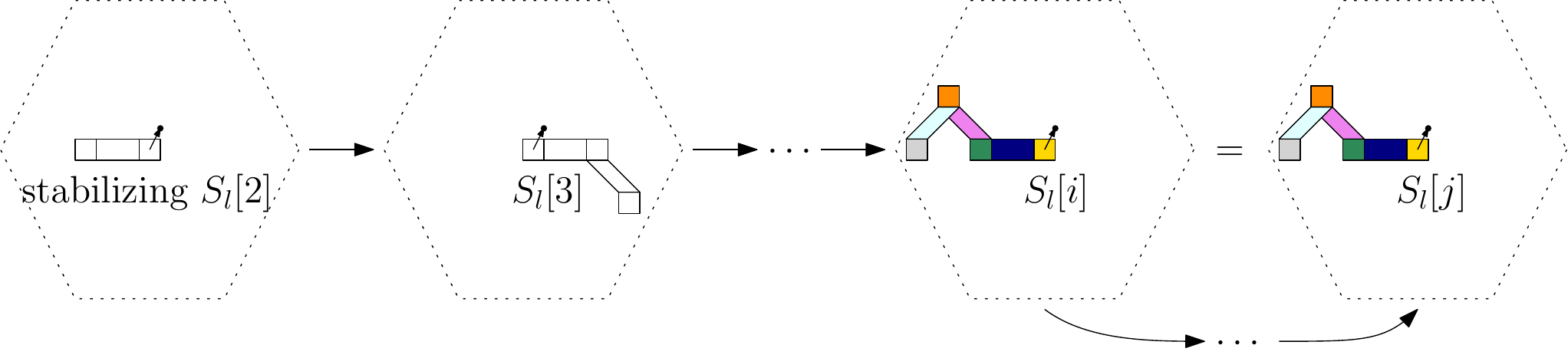}
		\caption{A series of event horizons (in dotted hexagons) for the first bead of $S_l[i]$. Shapes~$S_l$'s and $S_p$'s are simplified for better readability. There exists $1\le i<j\le 5^{4p_l+4p_p}{+}1$ such that they have exactly same event horizon for beads within $S_p[i]$ ($S_l[i]$) and $S_p[j]$ ($S_l[j]$).}
		\label{fig_koch1}
	\end{figure}
	
	Now we assume that $3l+3d+2\le\delta<12l+12d+11$, which is the distance between two $S_p$'s apart by $8$ unit distances in the lattice for the Koch curve. Suppose that we want to stabilize beads in $S_{pl}[i]$. The event horizon for beads in $S_{pl}[i]$ covers only $S_p$'s and $S_l$'s which are within $8$~unit distances from $S_p[i]$ in the lattice for the Koch curve. Similar to the previous case, we can find at most two $q_i$'s of distance~$9$, and all $S_p$'s and $S_i$'s are empty except those who represent the curve from $q_i$'s. Due to the length constraint, there can be at most $32$ $S_p$'s and $S_i$'s that overlap with the event horizon. Thus, among $S_{pl}[1]$ to $S_{pl}[5^{32p_l+32p_p}+1]$, there should exist $S_{pl}[i]$ and $S_{pl}[j>i]$ that have exactly the same event horizon for points within, and we know that the second assumption is also wrong. Similarly, we can expand the proof for arbitrarily large $\delta$. 
	Therefore, we know that for any given $\delta$, there is no delay-$\delta$ deterministic OS that can draw the Koch curve.
	
	Note that we can remove the third assumption about the OS period. If a period~$p$ is not $p_p+p_l$, beads in $S_{pl}[i]$ and $S_{pl}[i{+}pj]$ are exactly the same. Thus, for instance, if we assume that $\delta<3l+3d+2$, among $S_{pl}[1]$ to $S_{pl}[p\cdot 5^{4p_l+4p_p}{+}1]$, there should exist $S_{pl}[i]$ and $S_{pl}[j>i]$ that have exactly the same event horizon and produce the same conformation.
\end{proof}

\section{Impossibility of the Infinite Minkowski Curve}
\label{sec_min}

We study another example of infinite fractal curves---the Minkowski curve.
The Minkowski curve starts from a segment, then recursively alternates each line segment as follows:
\begin{enumerate}
	\item Divide the line segment into four segments (we call these segment 1 to 4 from the start) of equal length.
	\item Draw a square with segment 2 as a side to the left of the original segment, and the other square with segment 3 as a side to the right.
	\item Remove segments 2 and 3.
\end{enumerate}
Using the L-system, the Minkowski curve can be encoded as follows:
\begin{itemize}
	\item Variable: $F$
	\item Constants: $+,-$
	\item Axiom: $F$
	\item Production Rule: $F\to F+F-F-FF+F+F-F$,
\end{itemize}
where $F$ denotes a segment, $-$ denotes $90^\circ$ right turn and $+$ denotes $90^\circ$ left turn. \rfig{fig_min_ex} illustrates the Minkowski curve after three iterations.

\begin{figure}[htb]
	\centering
	\includegraphics[width=0.6\textwidth]{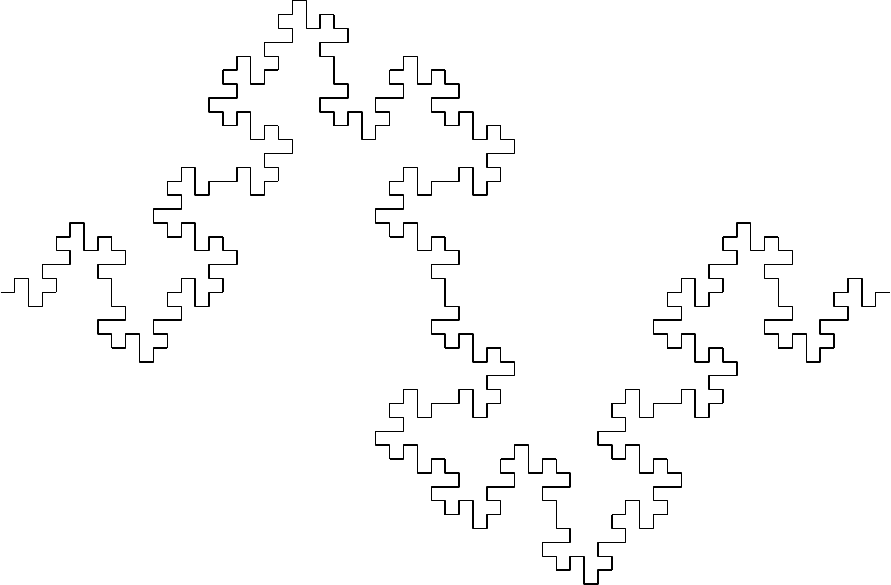}
	\caption{The Minkowski curve after three iterations}
	\label{fig_min_ex}
\end{figure}

For the Minkowski curve case, since the OS works on the triangular lattice, we assume that the Minkowski curve is slanted to fit into the triangular lattice---a square in the square lattice is mapped into a rhombus in the triangular lattice. We call the lattice the rhombus lattice. We assume the followings for drawing the Minkowski curve:
\begin{enumerate}
	\item The (tilted) Minkowski curve consists of an infinite sequence of alternating points and segments on the rhombus lattice~$\Lambda_r$. We use a parallelogram~$S_l$ of width~$d$ and length~$l$ to represent a segment on the curve, and a rhombus~$S_p$ of side length~$d$ to represents a point on the curve (See \rfig{fig_min_shape}). The OS starts with covering the first $S_p$, and denote the $i$th $S_p$ ($S_l$) by $S_p[i]$ ($S_l[i]$).

	For convenience, we set up a coordinate~$(x,y)$ for the points of the rhombus lattice based on an arbitrary origin, where the unit $x$ vector heads right and the unit $y$ vector heads upper right. Based on the coordinate, we use $S_l(x,y,x',y')$ to represent an $S_l$ that starts from the point~$(x,y)$ and the direction is given by the vector~$(x',y')$. We also use $S_p(x,y)$ to represent an $S_p$ at the point~$(x,y)$ (See \rfig{fig_min_shape}).

	\begin{figure}[htb]
		\centering
		\includegraphics[width=0.6\textwidth]{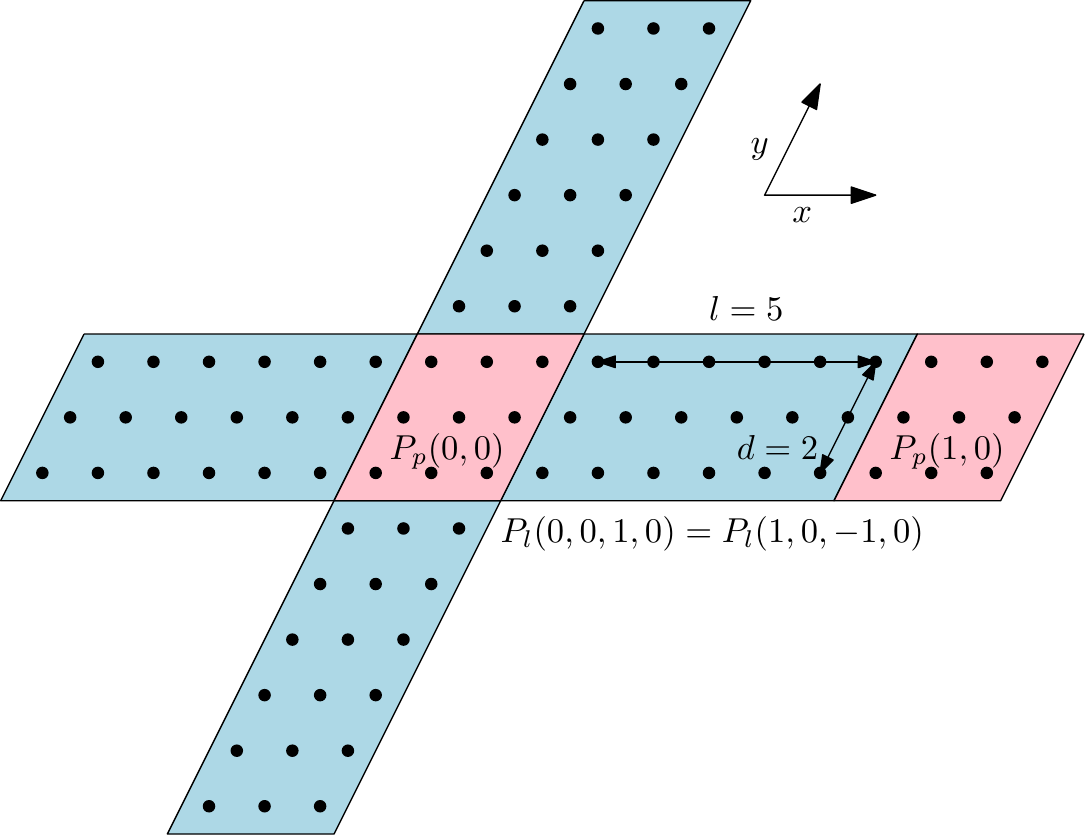}
		\caption{Shapes used to draw the Minkowski curve, and their aliases}
		\label{fig_min_shape}
	\end{figure}

	\item We use constant numbers of beads for a connecting line and a point---$p_p$~beads for $S_p$, and $p_l$~beads for $S_l$.
	\item The OS period~$p$ is $p=p_p+p_l$.
\end{enumerate}

Under these assumptions, we claim the following theorem.
\begin{theorem}
	\label{thm_min}
	There is no deterministic OS that can draw the tilted Minkowski curve.
\end{theorem}
\begin{proof}
	Similar to the proof for \rthm{thm_koch}, we fist assume that $\delta<3l+3d+5$. Suppose  we want to transcribe beads in $S_p[i]=S_p(x,y)$ and $S_l[i]=S_l(x,y,0,1)$, which represent the segment going up from the point~$(x,y)$. We have three cases for the previous segment:
	\begin{enumerate}
		\item $S_l[i{-}1]=S_l(x,y,-1,0)$, where the previous segment came from the left.
		\item $S_l[i{-}1]=S_l(x,y,1,0)$, where the previous segment came from the right.
		\item $S_l[i{-}1]=S_l(x,y,0,-1)$, where the previous segment came from the low.
	\end{enumerate}
	
	Considering all three cases, we can draw the event horizon~$E(i)$ for all beads in $S_{pl}[i]$ as in \rfig{fig_min_eh1}:
	
	\begin{figure}[htb]
		\centering
		\includegraphics[width=0.7\textwidth]{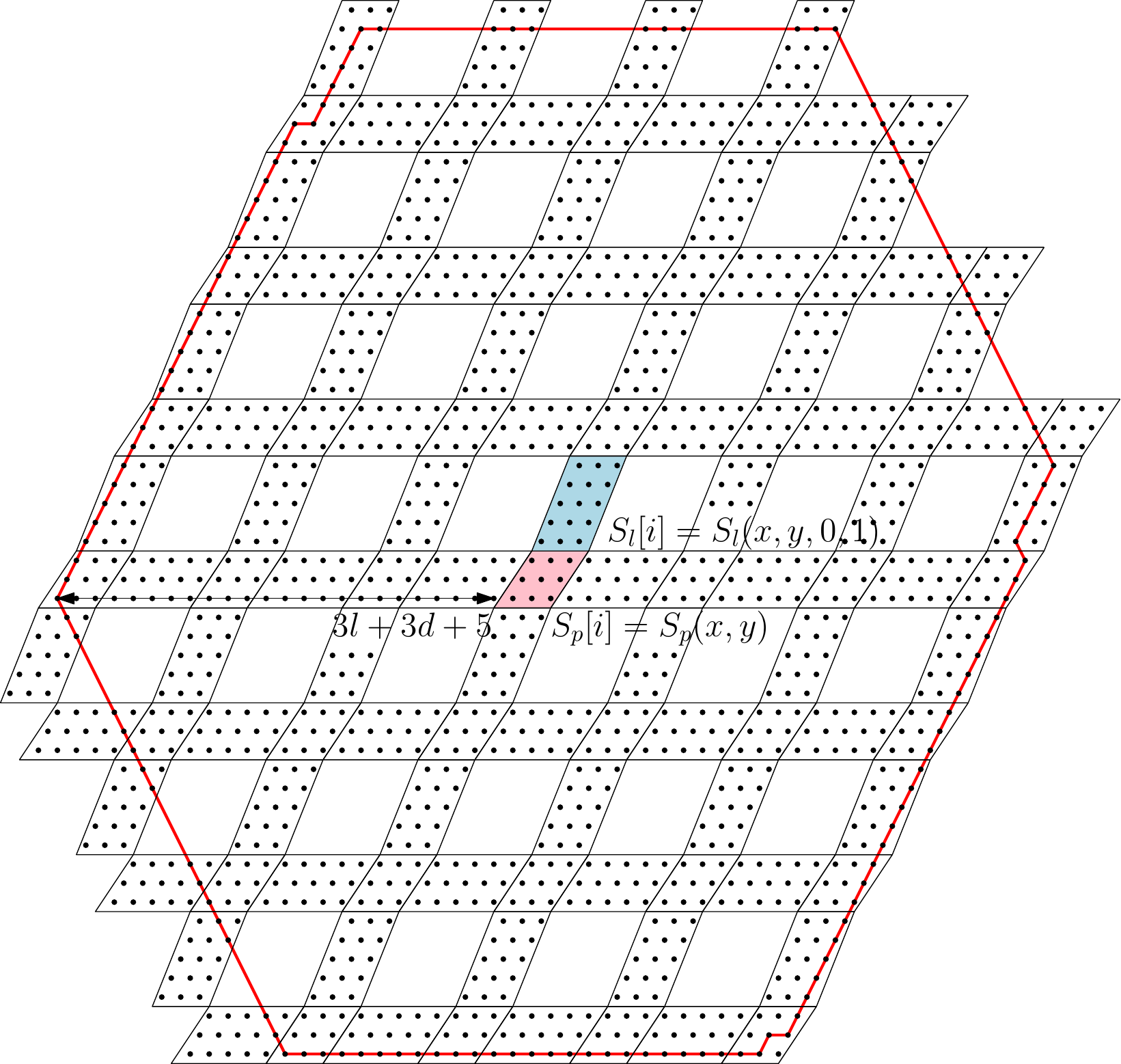}
		\caption{The event horizon~$E(i)$ for beads in $S_{pl}[i]$ is drawn in a red polygon.}
		\label{fig_min_eh1}
	\end{figure}
	
	We observe the property of the curve based on self-similarity. For a given~$n$, let $q_1$ be the starting point of a periodic substructure of length~$8^n$, and $q_2$ be the ending point. Points~$q_1$ and $q_2$ should be $4^n$ distance away. Now, suppose we draw the square of size~$2\cdot 4^n$ with the center~$q_2$. Then, the point~$q_3$ that appears first in the square (including the boundary) is at most $7\cdot 8^n$~segments away from $q_1$ (See \rfig{fig_min_touch} (a)). Moreover, the point~$q_4$ that appears last in the square is at most $8\cdot 8^n$~segments away from $q_1$ (See \rfig{fig_min_touch} (b)).
	
	\begin{figure}[htb]
		\centering
		\includegraphics[scale=1.0]{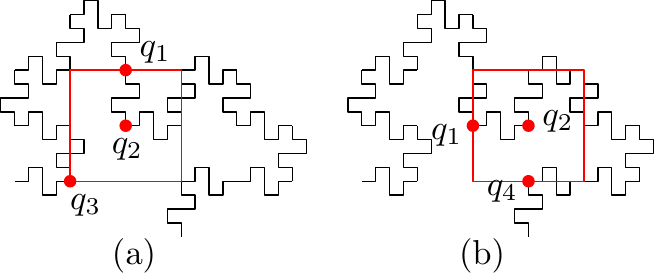}
		\caption{Properties of the curve when $n=1$. (a) $q_3$ is $56$~segments away from $q_1$. (b) $q_4$ is $64$~segments away from $q_1$.}
		\label{fig_min_touch}
	\end{figure}
	
	We have eleven different cases for the beads within $S_p$'s and $S_l$'s that overlap with $E(i)$ (See \rfig{fig_min_01}). In each case, shapes~$S_p[i]$ and $S_l[i]$ are distinguished by points within the shape. The thick rhombi represent $q_2$'s for $n=1$. We can observe that in all cases, the union of event horizons is covered by at most two squares with different $q_2$'s, whose shapes range from $S_p[i{-}5]$ to $S_p[i{+}6]$. According to the observation, we can say that a partial conformation in $S_{pl}[i]$ is dependent to beads in the sequence of shapes from $S_p[i{-}61]$ to $S_l[i{-}1]$. Thus, we have $i$ and $j$ such that $1\le i<j\le 1+5^{61p_l+61p_p}$ and $S_{pl}[i]$ and $S_{pl}[j]$ have exactly the same conformation. Since this results in a periodic sequence of segment turns, our assumption $\delta<3l+3d+5$ is wrong.
	
	Now we assume that $3l+3d+5\le\delta<15l+15d+29$. Suppose that we want to stabilize beads in $S_{pl}[i]$. The event horizon covers only $S_p$'s and $S_l$'s that are within the square of size $32$~unit distances with the center~$S_p[i]$. Similar to the previous case, we can find at most two~$q_2$'s with $n=2$ where the corresponding squares completely cover the event horizon. The earliest $S_p$ that can be such $q_2$ is $S_p[i{-}64]$, and we can say that a partial conformation in $S_{pl}[i]$ is dependent to beads in the sequence of shapes from $S_p[i{-}512]$ to $S_l[i{-}1]$. Thus, we have $i$ and $j$ such that $1\le i<j\le 1+5^{512p_l+512p_p}$ and $S_{pl}[i]$ and $S_{pl}[j]$ have exactly the same conformation. In a similar way, we can extend the proof for arbitrarily large~$\delta$. Similar to the Koch curve case, we can also remove the assumption about the OS period.

	\begin{figure}[htb]
		\centering
		\includegraphics[width=0.8\textwidth]{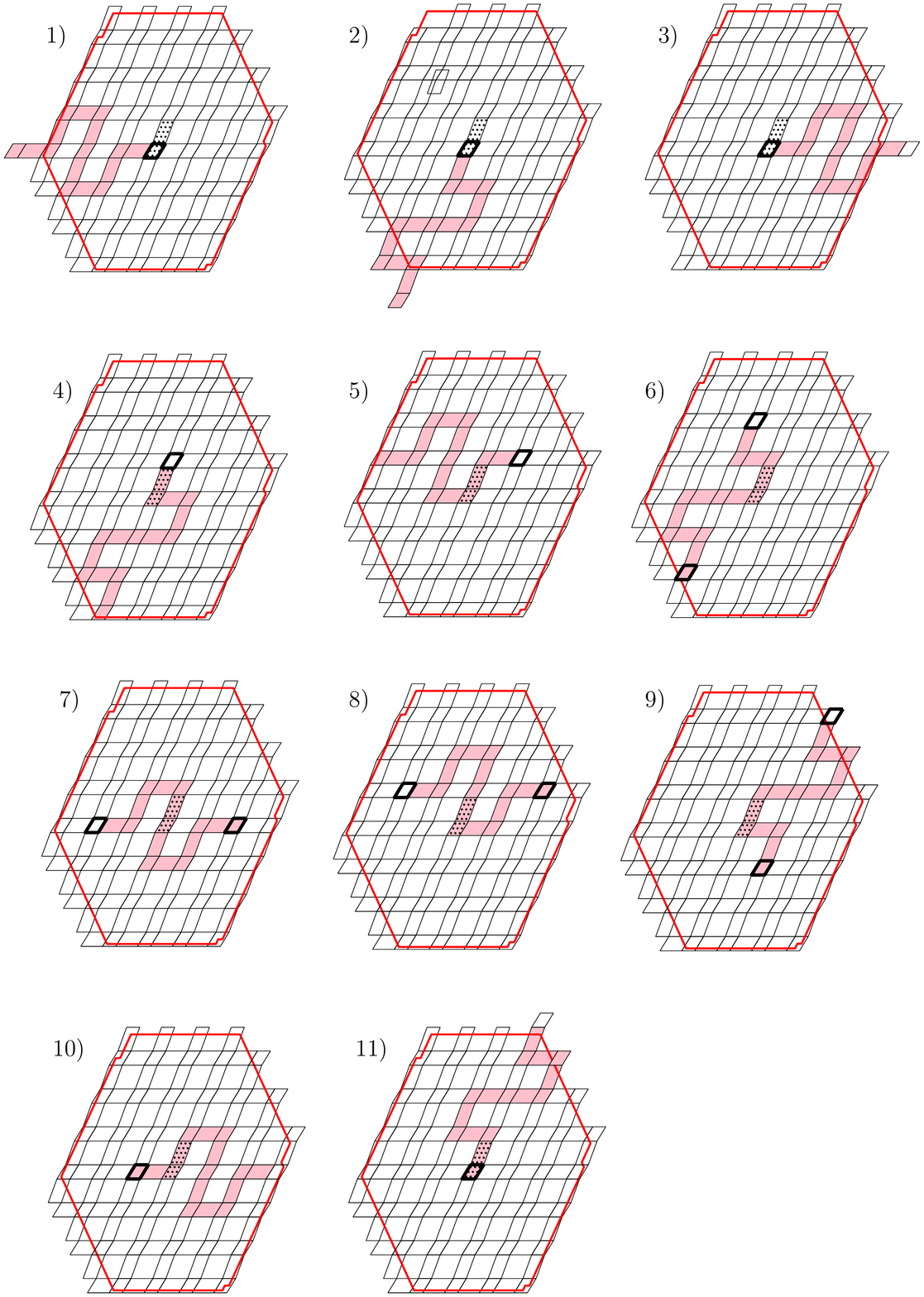}
		\caption{Eleven cases of the beads in the union of event horizons. }
		\label{fig_min_01}
	\end{figure}
	
\end{proof}

\section{Impossibility of Infinite Aperiodic Curves}
\label{sec_imp}

We expand the results from Sections~\ref{sec_koch} and \ref{sec_min} to impossibility of infinite aperiodic curves, which includes fractals. Construction of infinite periodic curves using a cyclic OS seems to be reasonable if we can design a partial OS that folds one period of the curve, and we have one running example---the glider in \rsec{sec_pre}. On the other hand, for infinite aperiodic curves, we propose sufficient conditions that makes curves impossible to fold. We make the following assumptions:
\begin{itemize}
	\item The curve is on an arbitrary lattice~$\Lambda$, and each point (segment) in $\Lambda$ is mapped to a shape~$S_p$ ($S_l$) in the triangular lattice~$\Lambda_o$.
	\item The OS uses $p_p$ ($p_l$) beads for $S_p$ ($S_l$). We say $p_{pl}=p_p+p_l$.
	\item The OS has the period of $p_o$.
	\item The curve is represented by an infinite alternation of $S_p$ and $S_l$, starting from $S_p[1]$. For convenience, we refer to the union of $S_p[i]$ and $S_l[i]$ as $S_{pl}[i]$.
\end{itemize}
Let $\delta(n)$ be an upper bound of the delay~$\delta$, dependent to a given integer~$n$. We propose the condition that curves are not foldable when $\delta\le\delta(n)$, and expand the result to all possible delays. Suppose we want to stabilize beads in $S_{pl}[i]$. Then, we have the maximum event horizon~$E(i,n)$ for beads in $S_{pl}[i]$, which is the union of all event horizons of radius~$\delta(n)+1$ whose centers are points in $S_{pl}[i]$ and points neighboring $S_p[i]$ in $S_l$'s adjacent to $S_p[i]$. Now, for each $i$, we have $S_p[r_{i,n}]$ that appears first in $E(i,n)$. Let 
$D_{i,n}=\max_{1\le j\le i} (j-r_{j,n})$ to be the maximum difference between $j$ and $r_{j,n}$ for all $j\le i$. Then, it takes $1+gcd(p_o,p_{pl})\cdot 5^{D_{i,n}p_{pl}}$ to have exactly the same conformation for the previous $D_{i,n}$~beads, which results in the same maximum event horizon and the same conformation for two shapes~$S_{pl}[j]$ and $S_{pl}[k>j]$. After $S_{pl}[j]$ and $S_{pl}[k]$, it is assured that previous $D_{i,n}$~beads have exactly the same conformation for the consecutively following shapes, and these shapes fold exactly the same. Since $D_{i,n}$ is dependent on $i$, we have the following theorem.

\begin{theorem}
	\label{thm_imp1}
	If there exists $i$ such that $1+gcd(p_o,p_{pl})\cdot 5^{D_{i,n}p_{pl}}\le i$, then it is impossible to draw a given infinite aperiodic curve with a cyclic OS whose delay is less than or equal to $\delta(n)$ and period is $p_o$.
\end{theorem}

In practice, the delay of the OS is bounded by the transcript length. If we consider a cyclic OS that has an infinite transcript, the delay can be arbitrarily large. We extend \rthm{thm_imp1} for arbitrarily large delays and obtain the following statement.

\begin{theorem}
	\label{thm_imp2}
	Suppose there exists a function~$\delta(n)$, where for all $n\ge 1$, there exists $i$ such that $1+gcd(p_o,p_{pl})\cdot 5^{D_{i,n}p_{pl}}\le i$. Then, it is impossible to draw a given aperiodic infinite curve with a cyclic OS whose period is $p_o$.
\end{theorem}

If $D_{i,n}$ is dependent only on $n$ and not $i$, then we may use an arbitrarily large $i$ to satisfy the conditions for all $n$'s regardless of $gcd(p_o,p_{pl})$. 
When we have such a case, the following statement holds.

\begin{theorem}
	\label{thm_imp3}
	If there exists a function~$\delta(n)$ such that, for all $n$, $D_{i,n}$ is independent of $i$, then it is impossible to draw a given infinite aperiodic curve with a cyclic OS regardless of the delay and the period of the OS.
\end{theorem}

Note that both in the proofs for Theorems~\ref{thm_koch} and \ref{thm_min}, we have $\delta(n)$ that makes $D_{i,n}$ independent of $i$, which is the \rthm{thm_imp3} case.

\section{Conclusions}
The oritatami system (OS) is a computational model inspired by RNA cotranscriptional folding, where an RNA transcript folds upon itself while synthesized out of a gene. Since the OS is a geometric computation model, it is natural to consider the problem of constructing fractal curves using the OS. We have formally defined the drawing of the curve by an OS. Then we have proved that it is impossible to draw two infinite fractal curves (Koch curve and Minkowski curve) by a cyclic OS. Moreover, we have proposed sufficient conditions that make
the folding of general infinite curves impossible. Our conjecture is that all fractal curves made by edge replacements are not foldable. However, we cannot directly apply the same approach to all curves, especially to curves that touch themselves such as Heighway dragons. Thus it is open to develop  a new approach for these self-touching curves.

\section*{Acknowledgements}
Han was supported  by the  Basic Science   Research    Program    through    NRF    funded    by    MEST~(2015R1D1A1A01060097). This work is partially supported by NIH R01GM109459, and by NSF's CCF-1526485 and DMS-1800443. This research was also partially supported by the Southeast Center for Mathematics and Biology, an NSF-Simons Research Center for Mathematics of Complex Biological Systems, under National Science Foundation Grant No. DMS-1764406 and Simons Foundation Grant No. 594594.

\bibliographystyle{abbrv}
\bibliography{draft}

\end{document}